\documentclass[11pt]{article}

\usepackage{hyperref,fullpage,amsmath,amsfonts,xspace,amssymb,amsthm}

\newtheorem{Thm}{Theorem}
\newtheorem{Lem}[Thm]{Lemma}
\newtheorem{Cor}[Thm]{Corollary}

\newtheorem{Def}{Definition}
\newtheorem{Fact}{Fact}

\newcommand\mbR{\mbox{$\mathbb{R}$}}
\newcommand\mbC{\mbox{$\mathbb{C}$}}
\newcommand\mcX{\mathcal{X}}
\newcommand\mcY{\mathcal{Y}}
\newcommand\mcW{\mathcal{W}}
\newcommand\mcZ{\mathcal{Z}}
\newcommand\h{\mathcal{H}}

\newcommand {\ie} {\textit{i.e.}\xspace}
\newcommand {\st} {\textit{s.t.}\xspace}

\newcommand\defeq{\stackrel{\mathrm{\scriptsize def}}{=}}
\newcommand\pr{\mbox{\bf Pr}}
\newcommand\av{\mbox{\bf{\bf E}}}

\newcommand\vecinv{\mathsf{vecinv}}

\newcommand\ket[1]{| #1 \rangle}
\newcommand\bra[1]{\langle #1 |}
\newcommand\qip[2]{\langle #1 | #2 \rangle}
\newcommand\suppress[1]{}

\newcommand\alice{\mbox{\sf Alice}\xspace}
\newcommand\bob{\mbox{\sf Bob}\xspace}
\newcommand\rcomm{\mbox{\sf {RComm}}\xspace}
\newcommand\qcomm{\mbox{\sf {QComm}}\xspace}

\newcommand\qcorr{\mbox{\sf {QCorr}}\xspace}

\newcommand\Q{\mbox{\sf {Q}}\xspace}
\newcommand\R{\mbox{\sf {R}}\xspace}
\newcommand\F{\mbox{\sf {F}}\xspace}

\newcommand\good{\mbox{\sf {Good}}\xspace}

\newcommand\rank{\mbox{\tt {rank}}\xspace}
\newcommand\srank{\mbox{\tt {S-rank}}\xspace}
\newcommand\prank{\mbox{\tt {rank}$_{\tt psd}$}\xspace}
\newcommand\tr{\mbox{\tt {tr}}\xspace}

\begin{document}

\begin{titlepage}

\title{\bf Correlation/Communication complexity of generating bipartite states}
\author{Rahul Jain\thanks{Department of Computer Science and Centre for Quantum Technologies, National University of Singapore, 2 Science Drive 3, Singapore 117542. Email: rahul@comp.nus.edu.sg}  \qquad
Yaoyun Shi\thanks{Department of Electrical Engineering and Computer Science, 2260 Hayward Street, University of Michigan, Ann Arbor, Michigan, 48109-2121, USA. Email: shiyy@eecs.umich.edu} \qquad
Zhaohui Wei\thanks{Centre for Quantum Technologies, National University of Singapore, 2 Science Drive 3, Singapore 117542. Email: cqtwz@nus.edu.sg} \qquad Shengyu Zhang\thanks{Department of Computer Science and Engineering and The Institute of Theoretical Computer Science and Communications, The Chinese University of Hong Kong. Email: syzhang@cse.cuhk.edu.hk}}
\date{}

\maketitle
\thispagestyle{empty}

\begin{abstract}
We study the correlation complexity (or equivalently, the communication complexity) of generating a bipartite quantum state $\rho$. When $\rho$ is a pure state, we completely characterize the complexity for approximately generating $\rho$ by a corresponding approximate rank, closing a 
gap left in Ambainis, Schulman, Ta-Shma, Vazirani and Wigderson ({\em SIAM Journal on Computing}, 32(6):1570-1585, 2003). When $\rho$ is a classical distribution $P(x,y)$, we tightly characterize the complexity of generating $P$ by the psd-rank, a measure recently proposed by Fiorini, Massar, Pokutta, Tiwary and de Wolf ({\em STOC} 2012). We also present a characterization of the complexity of generating a general quantum state $\rho$. 
\end{abstract}

\end{titlepage}

\pagestyle{plain}
\pagenumbering{arabic}

\section{Introduction}
In \cite{Zha12}, the following basic model was studied: Two parties, called \alice and \bob, aim to generate a target bipartite state $\rho  \in \h_A \otimes \h_B$ (Hilbert space $\h_A$ is in possession of \alice and Hilbert space $\h_B$ is in possession of \bob) using local quantum operations on a shared  {\em seed} state $\sigma \in \h_A \otimes \h_B$. The minimum {\em size}\footnote{The size of a quantum state $\sigma$ is defined to be half of the number of qubits of $\sigma$.} of this seed state is the \emph{quantum correlation complexity} of $\rho$, denoted $\Q(\rho)$. Since \alice and \bob can always just share $\rho$ itself, $\Q(\rho)$ is at most the number of qubits of $\rho$, so the correlation complexity is a sublinear complexity measure. Let $\{\ket{x} ~|~ x \in [\dim(\h_A)]\}$ be the {\em computational bases} for $\h_A$ and let $\{\ket{y} ~|~ y \in [\dim(\h_B)]\}$ be the computational bases for $\h_B$. We call a state $\rho$ {\em classical} if its eigenvectors are the computational basis states $\{\ket{x} \otimes \ket{y} ~|~ x \in [\dim(\h_A)],  y \in [\dim(\h_B)] \}$. Equivalently, it is just a classical probability distribution on the computational bases of $\h_A\otimes \h_B$. For a classical state $\rho$, the minimum size of a classical seed state is the \emph{randomized correlation complexity} of $\rho$, denoted $\R(\rho)$.  The work \cite{Zha12} exhibited a classical state $\rho$ of size $n$ with $\R(\rho) \geq \log_2(n) $ and $\Q(\rho) = 1$.

Above we considered the model in which \alice and \bob  start with some shared state $\sigma$ and produce target state $\rho$ by doing only local operations and no communication. On the other hand, we also consider the model in which \alice and \bob start with some tensor state $\sigma_A \otimes \sigma_B$ and do some local operations and communication and produce $\rho$ at the end of their protocol.  The {\em quantum communication complexity} of $\rho$, denoted $\qcomm(\rho)$, is defined as the minimum number of qubits exchanged between \alice and \bob, such that at the end of their protocol they output $\rho$. Again, when $\rho$ is classical, one can also define the {\em  randomized communication complexity} of $\rho$, denoted $\rcomm(\rho)$, as the minimum number of bits exchanged  between \alice and \bob, such that at the end of their protocol they output $\rho$. In \cite{Zha12} it is shown that for any classical state $\rho$,  
\begin{equation*} 
	\rcomm(\rho) = \R(\rho) = \lceil \log_2 \rank_+(P) \rceil,
\label{eq:rcorr*}
\end{equation*}
where $\rank_+(P)$ is the nonnegative rank\footnote{The nonnegative rank of a nonnegative matrix $A$ is the smallest number $r$ such that $A = \sum_{i=1}^r A_i$ where each $A_i$ is a \emph{nonnegative} rank-1 matrix.} of the $\dim(\h_A) \times \dim(\h_B)$ matrix $P$ with $P(x,y) \defeq (\bra{x} \otimes \bra{y}) \rho (\ket{x} \otimes \ket{y})$. It turns out that for a general quantum state $\rho$ it holds that $\qcomm(\rho) = \Q(\rho)$ as well. This fact was attributed to Nayak (personal communication) in \cite{Zha12}, and we shall see the reason in  a later section. 

We have considered above two extreme models. Instead we can also consider the intermediate model where \alice and \bob start with some shared state $\sigma$ and communicate between them to finally produce the target state $\rho$. In this case we  count the size of $\sigma$ plus the communication as the resource used towards the complexity. Let us denote $\tilde{\Q}(\rho)$ to be the minimum resource used by any protocol which produces $\rho$. It is clearly seen that $\qcomm(\rho) \leq \tilde{\Q}(\rho) \leq \Q(\rho)$, and hence $\qcomm(\rho) = \tilde{\Q}(\rho) = \Q(\rho)$ since $\qcomm(\rho) = \Q(\rho)$.


\subsection*{Our results} In this paper, we conduct more studies on the fundamental question of bipartite state generation. We consider approximate versions of $\Q(\rho)$ defined as follows. Below $\F(\rho,\rho')$ represents the {\em fidelity} between $\rho$ and $\rho'$. 
\begin{Def} \label{def:qrho}
Let $\epsilon > 0$.  Let $\rho$ be a quantum state  in $\h_A \otimes \h_B$. Define
$$ \Q_\epsilon(\rho)    \defeq \min\{ \Q(\rho')  ~| ~ \F(\rho, \rho') \geq 1 - \epsilon; ~ \rho' \in \h_A \otimes \h_B\} .$$
$$ \Q^{pure}_\epsilon(\rho)    \defeq \min\{ \Q(\ket{\phi}\bra{\phi})  ~| ~ \F(\rho, \ket{\phi}\bra{\phi}) \geq 1 - \epsilon; ~ \ket{\phi} \in \h_A \otimes \h_B\} .$$
\end{Def}
In \cite{ASTS+03}, Ambainis, Schulman, Ta-Shma, Vazirani and Wigderson showed that for any {\em pure state} $\ket{\psi} = \sum_{x,y} a_{x,y}\ket{x}\otimes \ket{y}$, 
\[\lceil \log_2 \rank_{2\epsilon}(A) \rceil \leq \Q^{pure}_\epsilon(\ket{\psi}\bra{\psi}) \leq \lceil \log_2 \rank_{\epsilon}(A) \rceil.\]
Above $A$ is the $\dim(\h_A) \times \dim(\h_B)$ matrix with $ A(x,y)= a_{x,y}$ and 
$$\rank_\epsilon(A) \defeq \min\{\rank(B) ~|~ \|A-B\|_2^2 \leq \epsilon \}.$$ 
Using Lemma \ref{lem:rankl2} (as mentioned in the next section), one can easily construct a state $\ket{\psi}\in \mbC^n\otimes \mbC^n$ such that $\rank_{2\epsilon}(A) = 1$ but $\rank_{\epsilon}(A) = n/2$, making the above two bounds arbitrarily far from each other. In this paper we show the following tight characterization.
\begin{Thm}\label{res:pure} Let $\epsilon > 0$. Let $\{\ket{x} ~|~ x \in [\dim(\h_A)]\}$ be the computational bases for $\h_A$ and let $\{\ket{y} ~|~ y \in [\dim(\h_B)]\}$ be the computational bases for $\h_B$. Let $\ket{\psi} = \sum_{x,y} a_{x,y} \ket{x} \otimes \ket{y}$. Let $A$ be defined as $A(x,y) = a_{x,y}$. Then \[\Q_\epsilon(\ket{\psi}\bra{\psi}) = \Q^{pure}_\epsilon(\ket{\psi}\bra{\psi}) = \lceil \log_2 \rank_{2\epsilon-\epsilon^2}(A) \rceil.\]
\end{Thm} 
Our result not only improves the bounds in~\cite{ASTS+03} to optimal, but also shows that allowing a \emph{mixed} state to approximate a pure state $\ket{\psi}$ does not help, for any $\ket{\psi}$ and any approximation ratio $\epsilon$.

\medskip
Our second result is for the case of a classical state $\rho$. Previously \cite{Zha12} gave upper and lower bounds: 
$$\frac{1}{4}\log_2 \rank(P) \leq \Q(\rho) \leq \min_{Q:\ Q\circ \bar Q = P} \log_2 \rank(Q).$$
Above $P$ is given by $P(x,y) = (\bra{x} \otimes \bra{y}) \rho (\ket{x} \otimes \ket{y})$ and $\circ$ is the Hadamard (i.e. entry-wise) product of matrices. How tight these bounds are is not clear yet, and an open question asked in \cite{Zha12} was a characterization of $\Q(\rho)$. In this paper, we answer this question by showing a tight characterization in terms of {\em psd-rank} of $P$, a concept recently proposed in \cite{FMP+12} by Fiorini, Massar, Pokutta, Tiwary and de Wolf. For a nonnegative matrix $P$, its psd-rank, denoted $\prank(P)$, is the minimum $r$ such that there are $r \times r$ positive semi-definite matrices $C_x$, $D_y$, satisfying that $P(x,y) = \tr ( C_x  D_y) $. We show the following result.
\begin{Thm}\label{res:distribution}
Let $\{ \ket{x} ~|~ x \in [\dim(\h_{A})]\}$ be the computational bases for $\h_A$ and let $\{ \ket{y}~ |~ y \in [\dim(\h_{B})]\}$ be the computational bases for  $\h_B$. Let 
$$ \rho =   \sum_{x\in  [\dim(\h_{A})] \atop y\in  [\dim(\h_{B})]} p_{x,y} \cdot \ket{x}\bra{x} \otimes \ket{y}\bra{y} \enspace .$$
Let $P$ be a $[\dim(\h_{A})] \times [\dim(\h_{B})]$ matrix with $P(x,y) = p_{x,y}$. Then $Q(\rho) = \lceil \log_2 \prank(P) \rceil$.
\end{Thm}
Along with the characterization $\R(\rho) =  \lceil \log_2 \rank_+(P) \rceil$ (shown in ~\cite{Zha12}), 
it is interesting to see that for classical states $\rho$, randomized correlation/communication complexity is all about nonnegative rank, and the quantum correlation/communication complexity is all about the psd-rank of the corresponding matrix $P$.

For a general quantum state $\rho$ we show the following characterization of $\Q(\rho)$.
\begin{Thm}\label{res:qrho}
Let $\rho$ be a quantum state  in $\h_A \otimes \h_B$.  Let $\{ \ket{x} ~|~ x \in [\dim(\h_{A})]\}$ be the computational bases for $\h_A$ and let $\{ \ket{y}~ |~ y \in [\dim(\h_{B})]\}$ be the computational bases for  $\h_B$. Then $\Q(\rho) = \lceil \log_2 r \rceil $ where $r$ is the minimum number such that there exist  matrices $\{A_x ~|~ x \in  [\dim(\h_{A})] \}$ and $\{B_y ~|~ y \in  [\dim(\h_{B})] \}$, each with $r$ columns, and 
\begin{align*}
\rho =  \sum_{x,x'\in  [\dim(\h_{A})]  \atop y,y'\in  [\dim(\h_{B})]} \ket{x}\bra{x'} \otimes \ket{y}\bra{y'} \cdot \tr \Big( (A_{x'}^\dag A_x)^T (B_{y'}^\dag B_y)  \Big) \enspace .
\end{align*}
\end{Thm}

\medskip 

\suppress{
Finally we give an upper bound on the classical complexity of approximating a classical state, viewed as probability distribution, in terms of the {\em common information} introduced by Wyner \cite{Wyn75} defined as below.
\begin{Def}[Wyner, \cite{Wyn75}] \label{def:comminf}
    The \emph{common information} $C(X:Y)$ between two random variables $X$ and $Y$ is the minimum value of $I(XY:W)$ (the mutual information between $XY$ and $W$), where the minimum is over all random variable $W$ such that $X$ and $Y$ are independent conditioned on $W$.
\end{Def}
We consider approximate version of classical complexity.
\begin{Def} Let $\epsilon >0$. Let $P$ be a probability distribution on $\mcX \times \mcY$. Define
$$\R_\epsilon(P) \defeq \{ \R(Q) ~|~ Q \text{ is a probability distribution on $\mcX \times \mcY$ with } \|P - Q\|_1 \geq 1 - \epsilon\} .$$
\end{Def}
For convenience we take the approximation by $\ell_1$ distance, $\|P-Q\|_1$, between $P$ and $Q$, in the definition above, instead of fidelity distance as in definitions of $\Q_\epsilon(\rho)$ and $\Q^{pure}_\epsilon(\rho)$.
\begin{Thm} \label{res:class} Let $\beta, \delta, \gamma >0$. Let $P$ be a distribution on  $\mcX \times \mcY$. Let random variables $XY$ be distributed according to $P$ such that $X$ is distributed in $\mcX$ and $Y$ is distributed in $\mcY$. Then,
    $$\R_{6\beta + \delta + 2\gamma}(X,Y) \leq C(X:Y)/\beta + 2\log(1/\delta) + \log\log(1/\gamma). $$ 
\end{Thm}

\medskip
}
The rest of the paper is organized as follows. In the next section we discuss our notation and some information theoretic preliminaries. In section \ref{sec:pure}, we prove Theorem \ref{res:pure}. In section \ref{sec:qrho} we prove Theorem \ref{res:distribution} and Theorem~\ref{res:qrho}. 

\section{Preliminaries}\label{sec:pre}

\subsection*{Matrix theory} For a natural number $n$ we let $[n]$ represent the set $\{1,2, \ldots, n\}$. For a matrix $A$, we let $A^T$ represent the transpose of $A$, $A^*$ represent the conjugate of $A$ and $A^\dag$ represent the conjugate transpose of $A$. An operator $A$ is said to be {\em Hermitian} if $A^\dag = A$. A Hermitian operator $A$ is said to be positive semi-definite if all its eigenvalues are non-negative. We will use the following fact.
\begin{Fact}
\label{fact:psd}
Let $\ket{v_1}, \ldots, \ket{v_r}$ be vectors in $\mbC^n$ for some $n \geq 1$. Then the $r \times r$ matrix $A$ defined by $A(i,j) \defeq \qip{v_i}{v_j}$ is positive semi-definite. 
\end{Fact}
If $A$ is positive semi-definite then so is $A^T = A^*$. We let $\sigma_1(A) \geq \cdots \geq \sigma_n(A)$ denote singular values of $A$. The rank of $A$, denoted $\rank(A)$, is defined to be the number of the non-zero singular values of $A$. The {\em Frobenius norm} of $A$ is defined as $\|A\|_2 = \sqrt{\sum_i \sigma_i(A)^2}$ and its {\em trace norm} is defined as $\|A\|_{1} = \sum_i \sigma_i$. For $\epsilon >0$, define 
$$ \rank_\epsilon(A) = \min \{\rank(B) ~|~ \|A-B\|_2^2 \leq \epsilon \}.$$ 
The following well-known result says that the best way to approximate $A$ (under the Frobenius norm) with the least rank is by taking the large singular values part.
\begin{Lem}[Eckart-Young, \cite{EY36}]\label{lem:rankl2}
    Let $\|A\|_2 = 1$ and $\epsilon >0$. Then,
    \[
    \rank_\epsilon(A) = \textup{the minimum } k \textup{ such that }  \sum_{i=1}^k \sigma_i(A)^2 \geq 1-\epsilon \enspace .
    \]
\end{Lem}

\suppress{
The next theorem exhibits a weak majorization of the singular values of the product of two matrices by the entry-wise product of the singular values of the two matrices. For a proof see, for example, the standard textbook \cite{HJ94} (Theorem 3.3.14).

\begin{Thm}\label{thm:SV-majorize}
 Let $A$ be an $n \times p$ matrix and let $B$ be a $p \times m$ matrix. Then,
 $$\sum_{i=1}^k \sigma_i(AB) \leq \sum_{i=1}^k \sigma_i(A) \sigma_i(B),$$ 
for all $k \in [\min\{n,p,m\}]$.
\end{Thm}
}

The following definition of psd-rank of a matrix was proposed in \cite{FMP+12}.
\begin{Def}[\cite{FMP+12}]
    For a matrix $P\in \mbR_+^{n\times m}$, its psd-rank, denoted $\prank(P)$, is the minimum number $r$ such that there are positive semi-definite matrices $C_x,D_y\in \mbC^{r\times r}$ with $\tr(C_x D_y) = P(x,y)$, $\forall x \in [n], y\in [m]$.
\end{Def}

\subsection*{Quantum computing}
A quantum state $\rho$ in Hilbert space $\h$, denoted $\rho \in \h$, is a trace one positive semi-definite operator acting on $\h$. The size of a state $\rho$ is defined to be half the number of qubits of $\rho$. Here we take the factor of half because we shall talk about a correlation as a \emph{shared} resource. It is consistent with the convention that when the two parties shares a classical correlation $(X,Y)$, where $Y = X = R$ for a $r$-bit random string $R$, we say that they share a random variable $R$ of size $r$. A quantum state is called pure iff it is rank one. We often also identity a pure state with its unique eigenvector with non-zero eigenvalue. For quantum states $\rho$ and $\sigma$, their fidelity is defined as $\F(\rho,\sigma) \defeq \tr(\sqrt{\sigma^{1/2}\rho\sigma^{1/2}})$. For $\rho, \ket{\psi} \in \h$, we have $\F(\rho,\ket{\psi}\bra{\psi}) = \sqrt{\bra{\psi} \rho \ket{\psi}}$. We define norm of $\ket{\psi}$ as $\|\ket{\psi}\| \defeq \sqrt{\qip{\psi}{\psi}}$. For a quantum state $\rho \in  \h_A \otimes \h_B$, we let $\tr_{\h_B} \rho$ represent the partial trace of $\rho$ in $\h_A$ after tracing out $\h_B$. Let $\rho \in \h_A$ and $\ket{\phi} \in \h_A \otimes \h_B$ be such that $\tr_{\h_B} \ket{\phi}\bra{\phi}  = \rho$, then we call $\ket{\phi}$ a {\em purification} of $\rho$. For a pure state $\ket{\psi} \in \h_A \otimes \h_B$, its {\em Schmidt decomposition}  is defined as $\ket{\psi} = \sum_{i=1}^r \sqrt{p_i} \cdot \ket{v_i} \otimes \ket{w_i}$, where $\{\ket{v_i} \in \h_A\}$ are orthonormal, $\{\ket{w_i} \in \h_B\} $ are orthonormal, $\forall i :~ p_i \geq 0$ with $\sum_{i=1}^r p_i = 1$. It is easily seen that $r$ is also equal to $\rank(\tr_{\h_A} \ket{\psi}\bra{\psi} ) =\rank(\tr_{\h_B} \ket{\psi}\bra{\psi} ) $ and is therefore the same in all Schmidt decompositions of $\ket{\psi}$.  This number is also referred to as the {\em Schmidt rank} of $\ket{\psi}$ and denoted $\srank(\ket{\psi})$. 
Sometimes we absorb the coefficients $\sqrt{p_i}$ in $\ket{v_i} \otimes \ket{w_i}$, in which case $\ket{v_i}, \ket{w_i}$ may not be unit vectors. The following is easily verified.
\begin{Fact}\label{fact:sranksame}
Let $U_A$ be a unitary operator on $\h_A$ and let $U_B$ be a unitary operator on $\h_B$. Let $\ket{\psi} \in \h_A \otimes \h_B$. Then $\srank(\ket{\psi}) = 
\srank((U_A \otimes U_B) \ket{\psi})$.
\end{Fact}
The following fact follows by considering Schmidt decomposition of the pure states involved; see, for example, Ex(2.81) of \cite{NC00}.
\begin{Fact}\label{fact:local}
Let $\ket{\psi}, \ket{\phi} \in \h_A \otimes \h_B$ be such that $\tr_{\h_B} \ket{\phi}\bra{\phi} = \tr_{\h_B} \ket{\psi}\bra{\psi} $. There exists a unitary operation $U$ on $\h_B$ such that $(I_{\h_A} \otimes U) \ket{\psi} = \ket{\phi}$, where $I_{\h_A}$ is the identity operator on $\h_A$.
\end{Fact}
The following fundamental fact is shown by Uhlmann~\cite{NC00}.
\begin{Fact}[Uhlmann, \cite{NC00}] \label{fact:uhlmann}
Let $\rho, \sigma \in \h_A$. Let $\ket{\psi} \in \h_A \otimes \h_B$ be a purification of $\rho$ and $\dim(\h_A) \leq \dim(\h_B)$. There exists a purification $\ket{\phi} \in \h_A \otimes \h_B$ of $\sigma$ such that $\F(\rho, \sigma) = |\qip{\phi}{\psi}|$.
\end{Fact}
We define the {\em approximate Schmidt rank} as follows.
\begin{Def}
\label{def:appsrank}
Let $\epsilon>0$. Let $\ket{\psi}$ be a pure state in $\h_A \otimes \h_B$. Define
$$\srank_\epsilon(\ket{\psi}) \defeq \min\{ \srank(\ket{\phi}) ~|~ \ket{\phi} \in \h_A \otimes \h_B \text{ and } \F(\ket{\psi}\bra{\psi}, \ket{\phi}\bra{\phi}) \geq 1 - \epsilon\} .$$
\end{Def}
Let $\{\ket{x} ~|~ x \in [\dim(\h_A)]\}$ be the computational bases for $\h_A$ and let $\{\ket{y} ~|~ y \in [\dim(\h_B)]\}$ be the computational bases for $\h_B$. We define linear map $\vecinv$ which takes vectors in $\h_A \otimes \h_B$ and maps them to operators from $\h_B$ to $\h_A$. For all $x \in [\dim(\h_A)], y \in [\dim(\h_B)]$ define  $\vecinv(\ket{x} \otimes \ket{y}) \defeq \ket{x}\bra{y}$ and extend to all vectors in $\h_A \otimes \h_B$ by linearity. For $\ket{\psi} \in \h_A \otimes \h_B$, it is easily seen that $\|\ket{\phi} \| = \|\vecinv(\ket{\psi})\|_2$.

In the following sections we assume Hilbert spaces $\h_A, \h_{A_1}, \h_{A_2}$ etc. are possessed by \alice and  Hilbert spaces $\h_B, \h_{B_1}, \h_{B_2}$ etc. are possessed by \bob.

We start by showing the following key lemma which we will use many times in the following sections. 
\begin{Lem}\label{lem:qrho}
Let $\rho$ be a quantum state  in $\h_A \otimes \h_B$. Then, 
\begin{equation*}
    \Q(\rho) = \min_{\h_{A_1}, \h_{B_1}}\{\big\lceil\log_2\big(\srank(\ket{\psi})\big)\big\rceil ~|~ \rho= \tr_{\h_{A_1} \otimes \h_{B_1}} \ket{\psi}\bra{\psi}\}.
\end{equation*}
\end{Lem}
\begin{proof}
Let $r \defeq \min_{\h_{A_1}, \h_{B_1}}\{\big\lceil\log_2\big(\srank(\ket{\psi})\big)\big\rceil ~|~ \rho= \tr_{\h_{A_1} \otimes \h_{B_1}} \ket{\psi}\bra{\psi}\}$. We first show $\Q(\rho) \leq r$. Let $\ket{\psi}$ be such that 
$$ r = \big\lceil\log_2\big(\srank(\ket{\psi})\big)\big\rceil \text { and }  \rho= \tr_{\h_{A_1} \otimes \h_{B_1}} \ket{\psi}\bra{\psi} .$$
Let $t \defeq {\srank(\ket{\psi})}$. Let $\ket{\psi}$ have a Schmidt decomposition 
$$\ket{\psi} = \sum_{i=1}^t \sqrt{p_i} \cdot \ket{v_i} \otimes \ket{w_i},$$ 
Let \alice and \bob start with the state 
$$\ket{\phi} = \sum_{i=1}^{t} \sqrt{p_i} \cdot \ket{i} \otimes \ket{i},$$
and transform $\ket{\phi}$ to $\ket{\psi}$ using local unitary transformations. This shows that $\Q(\rho) \leq \lceil \log_2 t \rceil =r$.

For  the other direction let $s \defeq \Q(\rho)$. Let \alice and \bob start with the seed state $\sigma$ and apply local {\em completely positive trace preserving maps} $\Phi_A, \Phi_B$ respectively to produce $\rho$. Let us assume without loss of generality that the number of qubits of $\sigma_A \defeq \tr_{\h_B} \sigma $ is at most $s$. Let $\sigma_A = \sum_{i=1}^{2^s} a_i \ket{v_i} \bra{v_i}$, where $a_i \geq 0$ is the $i$-th eigenvalue of $\sigma_A$ with eigenvector $\ket{v_i}$. Define
$$ \ket{\phi} \defeq \sum_{i=1}^{2^s} \sqrt{a_i} \cdot \ket{v_i} \otimes \ket{v_i} $$
and let 
$$ \ket{\phi'} = \sum_{i=1}^{2^s} \sqrt{a_i} \cdot \ket{v_i} \otimes \ket{w_i} $$
be a purification of $\sigma$, where $\forall i : \ \ket{v_i} \in \h_A \text{ and }  \ket{w_i} \in \h_B \otimes \h_{B_2}$.

Now consider the following operations by \alice and \bob. They start with the shared state $\ket{\phi}$. \bob using local unitary (after attaching ancilla $\ket{0}$ if needed) transforms $\ket{\phi}$ to $\ket{\phi'}$ (\bob can do this follows from Fact~\ref{fact:local}). \alice and \bob then simulate their maps $\Phi_A, \Phi_B$ on $\sigma$ by local unitaries   (each after attaching ancilla $\ket{0}$ if needed on their parts; such a simulation is a standard fact, please refer to~\cite{NC00}) and finally produce a purification $\ket{\theta} \in \h_A \otimes \h_{A_1} \otimes \h_B \otimes \h_{B_1}$ of $\rho$. Since \alice and \bob, using local unitary operations and attaching ancilla $\ket{0}$, transform $\ket{\phi}$ to $\ket{\theta}$, we have (using Fact~\ref{fact:sranksame}) $2^s \geq  \srank(\ket{\phi}) = \srank(\ket{\theta})$. This shows that $r \leq s$.
\end{proof}

The following lemma is credited to Nayak (personal communication) in~\cite{Zha12}.
\begin{Lem}
For a quantum state $\rho \in \h_A \times \h_B$ , $\Q(\rho) = \qcomm(\rho)$.
\end{Lem}
\begin{proof}
Clearly $\Q(\rho) \geq \qcomm(\rho)$. For the other direction let $r \defeq \qcomm(\rho)$. Let \alice and \bob start with the state $\sigma_A \otimes \sigma_B \in \h_A \otimes \h_B$, do local quantum operations, communicate $r$ qubits and at the end output $\rho$. This protocol can be converted into another protocol where \alice and \bob start with a purification $\ket{\phi} \in \h_A \otimes \h_{A_1} \otimes \h_B \otimes \h_{B_1}$ of $\sigma_A \otimes \sigma_B$ (with $\srank(\ket{\phi})=1$), do local unitaries, exchange $r$ qubits  and at the end output a purification $\ket{\psi} \in \h_A \otimes \h_{A_1} \otimes \h_B \otimes \h_{B_1}$ of $\rho$. Since local unitaries do not increase the Schmidt rank of the shared state and exchanging $r$ qubits increases the Schmidt rank by a factor at most $2^r$ (since the $\rank$ of the marginal state possessed by \alice increases by at most a factor $2$ on receiving a qubit from \bob, and similarly for \bob on receiving a qubit from \alice), we have $\srank(\ket{\psi}) \leq 2^r$. Hence from Lemma~\ref{lem:qrho}, $\Q(\rho) \leq r$.  
\end{proof}

\suppress{
\subsection*{Classical information theory} Let $P$ be a probability distribution on $\mcX$. For $x \in \mcX$, we let $P(x)$ denote the probability of $x$ under $P$. For $\mcX_1 \subseteq \mcX$, we define $P(\mcX_1) \defeq \sum_{x\in \mcX_1} P(x)$. The  {\em entropy} of $P$, denoted $H(P)$ is defined as $H(P) \defeq \sum_{x\in \mcX} - P(x) \log P(x)$. Let $P,Q$ be probability distributions on $\mcX$. Then the $\ell_1$    distance between them, denoted $\|P-Q\|_1$ is defined as $ \frac{1}{2}  \cdot \sum_{x\in \mcX}|P(x) - P(y)|$.  Let $P$ be a distribution on $\mcX \times \mcY$. For $(x,y) \in \mcX \times \mcY$, define $P(y|x) \defeq \frac{P(x,y)}{P(x)}$ and $P(y) \defeq \sum_{x\in\mcX} P(x,y)$. Similar one can define $P(x|y)$ and $P(x)$. We identify a random variable to also represent its distribution. Let random variables $XY$ be distributed in $\mcX \times \mcY$ and have joint distribution $P$. The {\em mutual information} between $X$ and $Y$ is defined as 
$$I(X:Y) \defeq H(X) + H(Y) - H(XY) = \sum_{(x,y) \in \mcX \times \mcY} P(x,y) \log \frac{P(x,y)}{P(x)\cdot P(y)}.$$ 
}

\section{Correlation complexity of approximating a pure state}\label{sec:pure}
In this section we prove Theorem~\ref{res:pure}. We start by first characterizing the approximate Schmidt rank.
\begin{Lem}
\label{lem:appsrank}
Let $\epsilon > 0$. Let $\ket{\psi}$ be a pure state in $\h_A \otimes \h_B$ with a Schmidt decomposition $\ket{\psi} = \sum_{i=1}^r \sqrt{p_i} \cdot \ket{v_i} \otimes \ket{w_i}$ (with $p_1 \geq p_2 \geq \ldots \geq p_r > 0$ and $\sum_{i=1}^r p_i = 1$). Let $r'$ be the minimum number such that $\sum_{i=1}^{r'} p_i  \geq (1 - \epsilon)^2$.  Then $r' = \srank_\epsilon(\ket{\psi})$.
\end{Lem}
\begin{proof}
We will first show that $r' \geq \srank_\epsilon(\ket{\psi})$. Let $q = \sum_{i=1}^{r'} p_i$. Define
$$ \ket{\phi} \defeq \frac{1}{\sqrt{q}} \cdot \sum_{i=1}^{r'} \sqrt{p_i} \cdot \ket{v_i} \otimes \ket{w_i} .$$
Then $ \F(\ket{\psi}\bra{\psi}, \ket{\phi}\bra{\phi})= |\qip{\psi}{\phi}| = \sqrt{q} \geq 1 - \epsilon$. Clearly $\srank(\ket{\phi}) =r'$ and hence $r' \geq \srank_\epsilon(\ket{\psi})$.

\vspace{0.1in}

Now we will show $r' \leq \srank_\epsilon(\ket{\psi})$. Let $s = \srank_\epsilon(\ket{\psi})$.  Let $\ket{\theta} \in \h_A \otimes \h_B$ be a pure state such that $ |\qip{\theta}{\psi}|= \F(\ket{\psi}\bra{\psi}, \ket{\theta}\bra{\theta}) \geq 1 - \epsilon$ and $\srank(\ket{\theta}) = s$.  Without loss of generality (by multiplying appropriate phase to $\ket{\theta}$) let us assume that $\beta \defeq \qip{\phi}{\theta}$ is real. Let  
$$\ket{\theta} = \sum_{j=1}^{s} \sqrt{q_i} \cdot \ket{v'_j} \otimes \ket{w'_j}$$
be a Schmidt decomposition of $\ket{\theta}$.  
Define
$$A \defeq \sum_{i=1}^r \sqrt{p_i} \cdot \ket{v_i}  \bra{w_i} \quad \text{  and  } \quad B \defeq \beta \cdot \sum_{i=1}^s \sqrt{q_i} \cdot \ket{v'_i}  \bra{w'_i}.$$
Note that $A = \vecinv(\ket{\psi})$ and $B = \vecinv(\ket{\theta'})$. Since $\{v_i\}$ and $\{w_i\}$ are orthonormal, $\{\sqrt{p_i}\}$ form the singular values of $A$. Similarly $\{\beta\cdot \sqrt{q_i}\}$ form the singular values of $B$. Now,
\begin{align*}
1 - (1-\epsilon)^2 & \geq  1 - \beta^2 \\
& =  \|\ket{\phi}\|^2 + \|\ket{\theta'}\|^2 - 2 \qip{\theta'}{\phi} =  \|\ket{\phi} - \ket{\theta'}\|^2 \\
& = \|\vecinv(\ket{\phi} -\ket{\theta'})\|^2_2 = \|\vecinv(\ket{\phi})  -\vecinv(\ket{\theta})\|^2_2 \\
& = \|A - B\|^2_2  . 
\end{align*}
Hence from Lemma~\ref{lem:rankl2}, $\srank(\theta) = \rank(B) \geq r'$.

\suppress{
Now,
\begin{equation*}
        1-\epsilon \leq |\qip{\psi}{\theta}| = |\tr(A^\dag B)| \leq \|A^\dag B\|_{1} = \sum_{i=1}^{\rank(A^\dag B)} \sigma_i(A^\dag B) \enspace .
    \end{equation*}
Consider   
\begin{align*}
1-\epsilon \leq \sum_{i=1}^{\rank(A^\dag B)} \sigma_i(A^\dag B) & \leq \sum_{i=1}^{\rank(A^\dag B)} \sigma_i(A^\dag) \sigma_i(B)  & (\text{from Theorem \ref{thm:SV-majorize}}) \\
& = \sum_{i=1}^s \sigma_i(A^\dag) \sigma_i(B) & (\text{since } \rank(B) = s) \\
& = \sum_{i=1}^s  \sqrt{p_i} \sqrt{q_i}   \\
& \leq \sqrt{\sum_{i=1}^s p_i}\sqrt{\sum_{i=1}^s q_i} & (\text{Cauchy-Schwartz inequality})\\
& = \sqrt{\sum_{i=1}^s p_i} \enspace .
\end{align*}
Hence $s \geq r'$. \qedhere
\suppress{Let $\{ \ket{u_i} ~|~ i \in [\dim(\h_A \otimes \h_B)]\}$ be an  orthonormal bases for $\h_A \otimes \h_B$ such that the first $r$ vectors in this set coincide with $\frac{1}{ \sqrt{\qip{v_i}{v_i} \cdot \qip{w_i}{w_i}}} \ket{v_i} \otimes \ket{w_i}$. 
Let  $\ket{\theta} = \sum_{j=1}^{s} \ket{v'_j} \otimes \ket{w'_j}$ be a Schmidt decomposition of $\ket{\theta}$. Then,
\begin{align*}
1-\epsilon & \leq |\qip{\theta}{\psi}|  = |\sum_{i=1}^r \sum_{j=1}^s \qip{v'_j}{v_i} \cdot \qip{w'_j}{w_i}| \\
& \leq \sqrt{\sum_{i=1}^r |\alpha_i|^2} \sqrt{\sum_{i=1}^r \qip{v_i}{v_i} \cdot \qip{w_i}{w_i}} \\
& \leq \sqrt{\sum_{i=1}^r \qip{v_i}{v_i} \cdot \qip{w_i}{w_i}}
\end{align*} 
}
}
\end{proof}

We can now get the desired characterization for  $\Q^{pure}_\epsilon(\ket{\psi}\bra{\psi}) $.
\begin{Thm} \label{thm:qpure}. 
Let $\epsilon > 0$. Let $\ket{\psi}$ be a pure state in $\h_A \otimes \h_B$ with a Schmidt decomposition $\ket{\psi} = \sum_{i=1}^r \sqrt{p_i} \cdot \ket{v_i} \otimes \ket{w_i}$ (with $p_1 \geq p_2 \geq \ldots \geq p_r > 0$ and $\sum_{i=1}^r p_i = 1$). Let $A = \sum_{i=1}^r \sqrt{p_i} \cdot \ket{v_i}  \bra{w_i} = \vecinv(\ket{\psi})$. Then,
    $\Q^{pure}_\epsilon(\ket{\psi}\bra{\psi}) = \lceil \log_2 \rank_{2\epsilon - \epsilon^2}(A)\rceil$.
\end{Thm}
\begin{proof} From the definitions and Lemma~\ref{lem:qrho}  it is clear that $\Q^{pure}_\epsilon(\ket{\psi}\bra{\psi})  = \lceil \log_2 \srank_\epsilon(\ket{\psi}) \rceil$. Also from Lemma~\ref{lem:rankl2} and Lemma~\ref{lem:appsrank} it follows that $\srank_\epsilon(\ket{\psi}) = \rank_{2\epsilon - \epsilon^2}(A)$ (by noting that $\{\sqrt{p_i}\}$ form singular values of $A$).
 \end{proof}
\suppress{
\begin{proof}
    We shall show that to generate $\ket{\psi}$, it is sufficient and necessary to use the large eigenvalue part of the matrix $A$. To be more precise, suppose $A = \sum_i \sigma_i \ket{u_i} \bra{v_i}$ is its singular value decomposition, with $\sigma_1 \geq \cdots \geq \sigma_n$. Let $k$ be the minimum value \st $\sum_{i=1}^k \sigma_i^2 \geq (1-\epsilon)^2$. We have that $k = \rank_{2\epsilon - \epsilon^2}(A)$ by Lemma \ref{lem:rankl2} and the simple calculation $\sigma_{k+1}(A)^2 + \cdots + \sigma_n(A)^2 \leq 1-(1-\epsilon)^2 = 2\epsilon-\epsilon^2$. Now we can show the conclusion as follows.

    \emph{Upper bound}: \alice and \bob share the state  \[\ket{seed} = \frac{\sum_{i=1}^k\sigma_i\ket{i} \ket{i}}{\sqrt{\sum_{i=1}^k \sigma_i^2}}.\]
    Then \alice rotates $\ket{i}$ to $\ket{u_i}$, and \bob rotates $\ket{i}$ to $\ket{v_i}$, resulting in the state $\ket{\psi'} = \frac{\sum_{i=1}^k\sigma_i\ket{u_i} \ket{v_i}}{\sqrt{\sum_{i=1}^k \sigma_i^2}}$. It is easily seen that
    \[\qip{\psi}{\psi'} = \frac{\sum_{i=1}^k \sigma_i^2}{\sqrt{\sum_{i=1}^k \sigma_i^2}} = \sqrt{\sum_{i=1}^k \sigma_i^2} \geq 1-\epsilon.\]

    \emph{Lower bound}: For any $\ket{\psi'} = \sum_{x,y} b_{x,y}\ket{x,y}$ with $|\qip{\psi}{\psi'}| \geq 1-\epsilon$, let $B = [b_{x,y}]$ and we have
    \begin{equation}
        1-\epsilon \leq |\qip{\psi}{\psi'}| = |\tr(A^\dag B)| \leq \|A^\dag B\|_{tr} = \sum_{i=1}^n \sigma_i(A^\dag B)
    \end{equation}
    By Theorem \ref{thm:SV-majorize}, $\sigma(A^\dag B)$, the sequence of singular values (in the non-increasing order) of $A^\dag B$ is weakly majorized by the entry-wise product of $\sigma(A^\dag)$ and that of $\sigma(B)$. Therefore by the definition of weak majorization, it holds that

\begin{align}
\sum_{i=1}^n \sigma_i(A^\dag B) & \leq \sum_{i=1}^n \sigma_i(A^\dag) \sigma_i(B) \\
& = \sum_{i=1}^k \sigma_i(A^\dag) \sigma_i(B) & (\text{since } \rank(B) = k) \\
& \leq \sqrt{\sum_{i=1}^k \sigma_i(A^\dag)^2}\sqrt{\sum_{i=1}^k \sigma_i(B)^2} & (\text{Cauchy-Schwartz Inequality})\\
& = \sqrt{\sum_{i=1}^k \sigma_i(A^\dag)^2} & (\|B\|_2 = 1 \text{ for pure state } \ket{\psi})
\end{align}
Thus $\sum_{i=1}^k \sigma_i(A)^2 \geq (1-\epsilon)^2$.
\end{proof}
}

The following lemma shows a monotonicity property for the approximate Schmidt rank.
\begin{Lem} \label{lem:srankmonotone}
Let $\epsilon >0$. Let $\ket{\psi}$ be a pure state in $\h_A \otimes \h_B$. Let $\ket{\theta}$ be a pure state in $\h_{A_1} \otimes \h_{B_1}$. Then,
    $$ \srank_\epsilon(\ket{\psi} \otimes \ket{\theta} ) \geq \srank_\epsilon(\ket{\psi}) .$$
Hence from Lemma~\ref{lem:qrho},
$$\Q^{pure}_\epsilon(\ket{\psi}\bra{\psi} \otimes \ket{\theta}\bra{\theta} ) \geq \Q^{pure}_\epsilon(\ket{\psi}\bra{\psi})  .$$
\end{Lem}
\begin{proof}
Let  $\ket{\psi} = \sum_{i=1}^{r} \sqrt{p_i} \cdot \ket{u_i^1} \otimes \ket{v_i^1}$ (with $p_1 \geq \cdots \geq p_r > 0$) and $\ket{\theta} = \sum_{i=1}^{s} \sqrt{q_i} \cdot \ket{u_i^2} \otimes \ket{v_i^2}$ be some Schmidt decompositions of  $\ket{\psi}$ and $\ket{\theta}$ respectively. Then
\begin{equation*}
    \ket{\psi} \otimes \ket{\theta} = \sum_{i,j} \sqrt{p_iq_j} \cdot \ket{u_i^1} \otimes \ket{u_i^2}\otimes \ket{v_i^1} \otimes \ket{v_i^2}.
\end{equation*}
Fix a minimal set $S\subseteq [r]\times [s]$ with $\sum_{(i,j)\in S} p_iq_j \geq 1-\epsilon$. Let $r' \defeq \srank_\epsilon(\ket{\psi}) $. We will show $|S| \geq r'$. Assume for contradiction $|S| \leq r'-1$. Let $S_1 = \{i ~|~ \exists j \text{ such that } (i,j)\in S\}$, then $|S_1| \leq |S| \leq r'-1$. We have
\begin{align*}
    \sum_{(i,j)\in S} p_iq_j & \leq \sum_{i\in S_1} p_i \leq p_1 + \cdots + p_{|S_1|} < 1-\epsilon,
\end{align*}
where the first inequality is because $\sum_{j: (i,j)\in S} q_j \leq 1$ for all $i$, the second inequality is because $p_i$'s are in the non-increasing order, and the last one is by the definition of $\srank_\epsilon(\ket{\psi}) = r'$, the smallest number such that $p_1+\cdots+p_{r'} \geq 1-\epsilon$ (from Lemma~\ref{lem:appsrank}). This contradicts the way we picked $S$ and hence
\begin{equation*}
    \srank_\epsilon(\ket{\psi} \otimes \ket{\theta}) = |S| \geq r' = \srank_\epsilon(\ket{\psi}). \qedhere
\end{equation*}
\end{proof}

\begin{Thm} \label{thm:pure} Let $\epsilon>0$. Let $\ket{\psi}$ be a pure state in $\h_A \otimes \h_B$.  Then,
    $\Q_\epsilon(\ket{\psi}\bra{\psi}) = \Q^{pure}_\epsilon(\ket{\psi}\bra{\psi})$.
\end{Thm}
\begin{proof}
By definition, we have $\Q_\epsilon(\ket{\psi}\bra{\psi}) \leq \Q^{pure}_\epsilon(\ket{\psi}\bra{\psi})$. Now consider the other direction. By the definition of $\Q_\epsilon(\ket{\psi}\bra{\psi})$,  there exists a $\rho \in \h_A \otimes \h_B$ such that 
\begin{equation} \label{eq:qrho}
\Q_\epsilon(\ket{\psi}\bra{\psi}) = \Q(\rho)  \text {  and   } \F(\rho, \ket{\psi}\bra{\psi}) \geq 1-\epsilon. 
\end{equation}
 By Lemma~\ref{lem:qrho}, there exists a purification $\ket{\phi}$ in $\h_A \otimes \h_{A_1} \otimes \h_B \otimes \h_{B_1}$ of $\rho$ with
\begin{equation} \label{eq:srank}
 \Q(\rho) = \lceil \log_2 \srank(\ket{\phi}\bra{\phi}) \rceil = \Q(\ket{\phi}\bra{\phi}) .
\end{equation}
Without loss of generality, we can assume that $\dim(\h_{A_1} \otimes \h_{B_1}) \geq \dim(\h_{A} \otimes \h_{B})$  (otherwise we can attach $\ket{0}$ to $\ket{\phi}$ appropriately). Now by Uhlmann's Theorem, there exists a purification $\ket{\psi'} \in \h_A \otimes \h_{A_1} \otimes \h_B \otimes \h_{B_1}$ of $\ket{\psi}\bra{\psi}$ such that $|\qip{\phi}{\psi'}| \geq 1-\epsilon$. Since $\ket{\psi}$ is a pure state, $\ket{\psi'} = \ket{\psi} \otimes \ket{\theta}$ for some $\ket{\theta} \in \h_{A_1} \otimes \h_{B_1}$. Therefore, 
\begin{align*}
    \Q^{pure}_\epsilon(\ket{\psi}\bra{\psi}) & \leq \Q^{pure}_\epsilon(\ket{\psi'}\bra{\psi'})  & (\text{from Lemma~\ref{lem:srankmonotone}})\\
& \leq \Q(\ket{\phi}\bra{\phi}) & (\text{from the definition of $\Q^{pure}_\epsilon(\ket{\psi'}\bra{\psi'})$}) \\
&= \Q(\rho) & (\text{from  Eq.~\eqref{eq:srank}}) \\
& = \Q_\epsilon(\ket{\psi}\bra{\psi}) \enspace .& (\text{from Eq.~\eqref{eq:qrho}})
\end{align*}
\end{proof}

Theorem~\ref{res:pure} now follows immediately by combining Theorem~\ref{thm:qpure} and Theorem~\ref{thm:pure} and noting that the matrix $A$ as defined in the statement of Theorem~\ref{res:pure} is $\vecinv(\ket{\psi})$.
\suppress{
\section{Classical distribution generation}\label{sec:distribution}
Let us first recall the definition of psd-rank, proposed in \cite{FMP+12}.
\begin{Def}
    For a matrix $M\in \mbR_+^{n\times n}$, its psd-rank $\prank(M)$ is the minimum integer $r$ such that there are psd matrices $A_i,B_j\in \mbC^{r\times r}$ with $\tr(A_i^\dag B_j) = M_{ij}$, $\forall i,j\in [n]$.
\end{Def}
In this section we will prove Theorem \ref{thm:distribution}, namely that $\Q(P)$ is completely characterized by the psd-rank of $P$. Before giving the formal proof, let us briefly mention some naive approach for the upper bound and why it does not work. Given psd matrices $A_x$ and $B_y$ with $\tr(A_x^\dag B_y) = p_{xy}$, it is natural to attempt to let \alice and \bob use $A_x$ and $B_y$ (properly normalized) respectively as their POVM measurements. This would work if $\sum_x A_x$ and $\sum_y B_y$ are both multiples of the identity matrix $I$, which is not guaranteed in general. So we need to find out a way to adjust the eigenvalues of $\sum_x A_x$.

\medskip
\begin{proof} (of Theorem \ref{thm:distribution})
    \emph{Lower bound}: Suppose the shared state is $\rho = \mu_i \sum_{i=1}^{2^q}\ket{\psi_i}\bra{\psi_i}$, where $q = 2r = 2\qcorr(p)$ and $\ket{\psi_i}$'s are pure states with Schmidt decomposition $\ket{\psi_i} = \sum_{j=1}^{2^r} \lambda_{ij} \ket{\psi_{ij}}\otimes \ket{\phi_{ij}}$. Suppose \alice and \bob apply POVM measurements $\{E_x\}$ and $\{F_y\}$, respectively. The probability of $(x,y)$ occurs is
\begin{align}
    p(x,y) & = \sum_{ijk} \mu_i \lambda_{ij}\lambda_{ik} \bra{\psi_{ij}} E_x \ket{\psi_{ik}} \cdot \bra{\phi_{ij}} F_y \ket{\phi_{ik}} = \sum_i \mu_i \tr(A_{i,x}^\dag B_{i,y})
\end{align}
where $(A_{i,x})_{jk} = \bra{\lambda_{ij}\psi_{ij}^*} E_x \ket{\lambda_{ik}\psi_{ik}^*}$, and $(B_{i,y})_{jk} = \bra{\phi_{ij}} F_y \ket{\phi_{ik}}$. Since $A_{i,x} \succeq 0$, $B_{i,y} \succeq 0$, and these matrices are of dimension $2^r \times 2^r$, we know that the $\prank(P) \leq 2^q2^r = 2^{3r}$.

    \medskip
    \emph{Upper bound}: Following the notation in the definition of \qcorr, we index the rows by $x$ and columns by $y$ in matrix $P$. By the definition of $\prank(P)$, there are $r\times r$ psd matrices $A_x$ and $B_y$ with $\tr(A_x^\dag B_y) = p_{xy}$. Let $A = \sqrt{r} \sum_x A_x$ and $B= \sqrt{r} \sum_y B_y$. Then $A \succeq 0$ and $B \succeq 0$, and thus they have spectral decompositions:
    \begin{equation}\label{eq: decomp of A and B}
        A = \sum_{i=1}^r a_i \ket{\alpha_i}\bra{\alpha_i}, \quad B = \sum_{i=1}^r b_i \ket{\beta_i}\bra{\beta_i}
    \end{equation}
    where $a_i\geq 0$, $b_i\geq 0$, and $\|\ket{\alpha_i}\|_2 = \|\ket{\beta_i}\|_2 = 1$.

    We first claim that actually all $a_i > 0$ and $b_i > 0$, namely $A$ and $B$ are full-rank. Suppose this is not true, for example, $A \ket{\alpha_i} = 0$ for some $i$. Then by the positivity of each $A_x$, it is easily seen that each $A_x\ket{\alpha} = 0$ as well. Now for each $y$, consider the spectral decomposition of $B_y$: $B_y = \sum_j b_{y,j} \ket{\beta_{y,j}} \bra{\beta_{y,j}}$ where $b_{y,j} \geq 0$. Define
    \begin{equation}
        \ket{\beta_{y,j}'} = \ket{\beta_{y,j}} - \qip{\alpha_i}{\beta_{y,j}} \ket{\alpha_i} \quad \text{and } \quad B_y' = \sum_j b_{y,j} \ket{\beta_{y,j}'} \bra{\beta_{y,j}'},
    \end{equation}
    then $B_y'\succeq 0$ by definition, and $\bra{\alpha_i} B_y' \ket{\alpha_i} = \sum_j b_{y,j}|\qip{\alpha_i}{\beta_y'}|^2$ = 0. One can also verify that
    \begin{align}
        \tr(A_x^\dag B_y') & = \tr(A_x \sum_j b_{y,j} \ket{\beta_{y,j}'} \bra{\beta_{y,j}'}) & (A_x \succeq 0) \\
        & = \sum_j b_{y,j} \bra{\beta_{y,j}'} A_x \ket{\beta_{y,j}'} \\
        & = \sum_j b_{y,j} \bra{\beta_{y,j}} A_x \ket{\beta_{y,j}} & (A_x^\dag\ket{\alpha_i} = A_x\ket{\alpha_i} = 0)\\
        & = \tr(A_x^\dag B_y) = p_{xy}.
    \end{align}
    Thus one can use $\{A_x\}$ and $\{B_y'\}$ to construct a psd decomposition of $P$ with dimension strictly smaller than $r$, contradictory to the assumption that $r$ is the optimum in the definition of $\prank(P)$. Indeed, one can use a unitary matrix $U$ to rotate $\ket{\alpha_i}$ to $\ket{1}$, and use $\{UA_xU^\dag\}$ and $\{UB_y'U^\dag\}$ for the psd decomposition. Note that $UA_xU^\dag \succeq 0$ and $\tr(UA_x^\dag U^\dag UB_y'U^\dag) = \tr(A_x^\dag B_y') = p_{xy}$, thus $\{UA_xU^\dag\}$ and $\{UB_y'U^\dag\}$ form a valid psd decomposition, but all matrices $UA_xU^\dag$ and $UB_y'U^\dag$ have the first row and first column being all zero. So we could have used the rest $r-1$ rows and columns to get a psd decomposition for $P$ with a strictly smaller dimension.

    Now that all $a_i$ and $b_j$ are strictly positive, we can define
    \begin{equation}
        R = \sum_{i=1}^r \frac{1}{\sqrt{a_i}} \ket{i}\bra{\alpha_i}, \quad S = \sum_{i=1}^r \frac{1}{\sqrt{b_i}} \ket{i}\bra{\beta_i}
    \end{equation}
    It is easy to calculate their inverse: $R^{-1} = \sum_{i=1}^r \sqrt{a_i} \ket{\alpha_i}\bra{i}$ and $S^{-1} = \sqrt{b_i} \ket{\beta_i}\bra{i}$.

    In the protocol, we will let \alice and \bob share the state
    \begin{equation}
        \ket{\psi'} = ((R^\dag)^{-1}\otimes (S^\dag)^{-1}) \ket{\psi}, \quad \text{where } \ket{\psi} = \frac{1}{\sqrt{r}} \sum_{i=1}^r \ket{i,i}.
    \end{equation}
    First, let us verify that it is indeed a quantum state; namely its $\ell_2$ norm is 1.
    \begin{align}
        & \quad \ \|((R^\dag)^{-1}\otimes (S^\dag)^{-1}) \ket{\psi} \|^2 \\
        & = \bra{\psi} (R^{-1}\otimes S^{-1}) ((R^\dag)^{-1}\otimes (S^\dag)^{-1}) \ket{\psi} \\
        & = \bra{\psi} ((\sum_i a_i \ket{\alpha_i}\bra{\alpha_i})\otimes (\sum_i b_i \ket{\beta_i}\bra{\beta_i})) \ket{\psi} & (\text{definitions of $R$ and $S$)}\\
        & = \bra{\psi} (A\otimes B) \ket{\psi} & (\text{Eq.\eqref{eq: decomp of A and B})}\\
        & = \bra{\psi} (\sqrt{r}\sum_x A_x\otimes \sqrt{r}\sum_y B_y) \ket{\psi} & (\text{definitions of $A$ and $B$)}\\
        & = r \sum_{x,y} \bra{\psi} (A_x\otimes B_y) \ket{\psi} \\
        & = \sum_{x,y} \sum_{i,j}\bra{i} A_x \ket{j} \bra{i} B_y \ket{j} & (\text{definition of } \ket{\psi})\\
        & = \sum_{x,y} \tr(A_x^\dag B_y) = \sum_{x,y}p_{xy} = 1
    \end{align}

    \alice uses the measurement $\{E_x\}$ and \bob uses $\{F_y\}$, where $E_x = \sqrt{r} R A_x R^\dag$ and $F_y = \sqrt{r} S B_y S^\dag$. They are indeed POVM measurements because it is not hard to verify that $A_x \succeq 0$ implies $R A_x R^\dag \succeq 0$, and
    \begin{equation}
        \sum_x \sqrt{r} R A_x R^\dag = R A R^\dag = \big(\sum_{i} \frac{1}{\sqrt{a_i}} \ket{i}\bra{\alpha_i}\big) \big(\sum_{i} a_i \ket{\alpha_i}\bra{\alpha_i}\big) \big(\sum_{i} \frac{1}{\sqrt{a_i}} \ket{\alpha_i}\bra{i}\big) = I
    \end{equation}
    Similar arguments hold for $\{F_y\}$. Finally, we need to check that the probability of $(x,y)$ occurs in the measurement is
    \begin{align}
        \pr[x,y] & = \bra{\psi'}(E_x\otimes F_y) \ket{\psi'} \\
        & = \bra{\psi} (R^{-1}\otimes S^{-1}) (\sqrt{r} R A_x R^\dag \otimes \sqrt{r} S B_y S^\dag) ((R^\dag)^{-1}\otimes (S^\dag)^{-1}) \ket{\psi} \\
        & = r\cdot \bra{\psi} (A_x\otimes B_y) \ket{\psi} = p_{xy},
    \end{align}
    as desired.
\end{proof}

If we define $\prank_\epsilon(P)$ as the minimum $\prank(P')$ for a nonnegative matrix $P'$ with $\|P-P'\| \leq \epsilon$ for some matrix norm $\|\cdot \|$, then an immediate corollary is that the complexity $\Q_\epsilon(P)$ of generating an $\epsilon$-approximation with respect to norm $\|\cdot \|$ is just $\log \prank_\epsilon(P)$.
\begin{Cor}
    $\Q_\epsilon(P) = \Theta(\log \prank_\epsilon(P))$.
\end{Cor}
}

\section{Correlation complexity of  a quantum state}
\label{sec:qrho}
In this section we show characterizations of correlation complexities for general quantum states and also for classical states and prove Theorem~\ref{res:distribution} and Theorem~\ref{res:qrho}. We start with the following lemma.
\begin{Lem} \label{lem:alt}
Let $\rho$ be a quantum state  in $\h_A \otimes \h_B$.  Let $\{ \ket{x} ~|~ x \in [\dim(\h_{A})]\}$ be the computational bases for $\h_A$ and let $\{ \ket{y}~ |~ y \in [\dim(\h_{B})]\}$ be the computational bases for  $\h_B$. There exists a purification $\ket{\psi}$ of $\rho$, with $\srank(\ket{\psi}) =r$, if and only if  there exist matrices $\{A_x ~|~ x \in  [\dim(\h_{A})] \}$ and $\{B_y ~|~ y \in  [\dim(\h_{B})] \}$, each with $r$ columns, such that   
\begin{align*}
\rho = \sum_{x,x' \in [\dim(\h_{A})] \atop ~y,y' \in [\dim(\h_{B})]} \ket{x}\bra{x'} \otimes \ket{y}\bra{y'} \cdot \tr \left((A_{x'}^\dag A_x)^T (B_{y'}^\dag B_y)\right)   \enspace .
\end{align*}
\end{Lem}
\begin{proof}
We first show the `only if' implication. Let $\ket{\psi}$ be a purification of $\rho$ in $\h_{A} \otimes \h_{A_1} \otimes \h_B \otimes \h_{B_1}$. Let $\srank(\ket{\psi}) =r$. Consider a Schmidt decomposition of $\ket{\psi}$.
\begin{align*}
\ket{\psi}  = \sum_{i=1}^r \ket{v^i} \otimes \ket{w^i} = \sum_{i=1}^r \left(\sum_{x\in  [\dim(\h_{A})]} \ket{x} \otimes \ket{v_x^i}\right) \otimes \left(\sum_{y\in  [\dim(\h_{B})]} \ket{y} \otimes \ket{w_y^i}\right) \enspace .
\end{align*}
Above for any $i, x, y$, the vectors $\ket{v^i}, \ket{w^i}, \ket{v_x^i}, \ket{w_y^i}$ are not necessarily unit vectors. Consider 
\begin{align*}
{\rho} &= \tr_{\h_{A_1} \otimes \h_{B_1}} \ket{\psi}\bra{\psi} \\
&= \tr_{\h_{A_1} \otimes \h_{B_1}}  \left(\sum_{i=1}^r \left(\sum_x \ket{x} \otimes \ket{v_x^i}\right) \otimes \left(\sum_y \ket{y} \otimes \ket{w_y^i}\right) \right) \left( \sum_{j=1}^r \left(\sum_{x'} \bra{x'} \otimes \bra{v_{x'}^j}\right) \otimes \left(\sum_{y'} \bra{y'} \otimes \bra{w_{y'}^j}\right) \right) \\
&= \tr_{\h_{A_1} \otimes \h_{B_1}}\sum_{i,j=1}^r \left(\sum_{x,x'\in  [\dim(\h_{A})]} \ket{x}\bra{x'} \otimes \ket{v_x^i}\bra{v_{x'}^j}\right) \otimes \left(\sum_{y,y'\in  [\dim(\h_{B})]} \ket{y}\bra{y'} \otimes \ket{w_y^i}\bra{w_{y'}^j}\right) \\
&= \sum_{i,j=1}^r \left(\sum_{x,x'\in  [\dim(\h_{A})]} \ket{x}\bra{x'} \cdot \tr(\ket{v_x^i}\bra{v_{x'}^j})\right) \otimes \left(\sum_{y,y'\in  [\dim(\h_{B})]} \ket{y}\bra{y'} \cdot \tr(\ket{w_y^i}\bra{w_{y'}^j})\right) \\
&= \sum_{x,x'\in  [\dim(\h_{A})]  \atop y,y'\in  [\dim(\h_{B})]} \ket{x}\bra{x'} \otimes \ket{y}\bra{y'} \left( \sum_{i,j=1}^r \qip{v_{x'}^j}{v_x^i} \cdot \qip{w_{y'}^j}{w_y^i} \right)   \enspace .
\end{align*}
For each $x \in [\dim(\h_A)]$, let us define matrices $A_x \defeq (\ket{v_x^1}, \ket{v_x^2}, \ldots, \ket{v_x^r})$. Similarly for each $y \in [\dim(\h_B)]$, let us define matrices $B_y \defeq (\ket{w_y^1}, \ket{w_y^2}, \ldots, \ket{w_y^r})$. Then from above,
\begin{align*}
\rho =  \sum_{x,x'\in  [\dim(\h_{A})]  \atop y,y'\in  [\dim(\h_{B})]} \ket{x}\bra{x'} \otimes \ket{y}\bra{y'} \cdot \tr \left((A_{x'}^\dag A_x)^T (B_{y'}^\dag B_y) \right)  \enspace .
\end{align*}

Next we show the `if' implication. Let there exist matrices $\{A_x ~|~ x \in  [\dim(\h_{A})] \}$ and $\{B_y ~|~ y \in  [\dim(\h_{B})] \}$, each with $r$ columns, such that   
\begin{align*}
\rho = \sum_{x,x' \in [\dim(\h_{A})] \atop y,y' \in [\dim(\h_{B})]} \ket{x}\bra{x'} \otimes \ket{y}\bra{y'} \cdot \tr \left((A_{x'}^\dag A_x)^T (B_{y'}^\dag B_y) \right)  \enspace .
\end{align*}
For $i \in [r]$, let $\ket{v_x^i}$ be the $i$-th column of $A_x$ and let $\ket{w_y^i}$ be the $i$-th column of $B_y$.
Define
\begin{align*}
\ket{\psi}  \defeq  \sum_{i=1}^r \left(\sum_x \ket{x} \otimes \ket{v_x^i}\right) \otimes \left(\sum_y \ket{y} \otimes \ket{w_y^i}\right) 
\end{align*}
It is clear that $\srank(\ket{\psi}) =r$. We can check, by analogous calculations as above, that
\begin{equation*}
\rho = \tr_{\h_{A_1} \otimes \h_{B_1}} \ket{\psi}\bra{\psi} \enspace .  \qedhere
\end{equation*}
\end{proof}

By combining Lemma~\ref{lem:qrho} and Lemma~\ref{lem:alt} we immediately get Theorem~\ref{res:qrho}. We now show Theorem~\ref{res:distribution} which we restate below for convenience.
\begin{Thm}\label{thm:qdist}
Let $\{ \ket{x} ~|~ x \in [\dim(\h_{A})]\}$ be the computational bases for $\h_A$ and let $\{ \ket{y}~ |~ y \in [\dim(\h_{B})]\}$ be the computational bases for  $\h_B$. Let 
$$ \rho =   \sum_{x\in  [\dim(\h_{A})]  \atop y\in  [\dim(\h_{B})]} p_{x,y} \cdot \ket{x}\bra{x} \otimes \ket{y}\bra{y} \enspace .$$
Let $P$ be a $[\dim(\h_{A})] \times [\dim(\h_{B})]$ matrix with $P(x,y) = p_{x,y}$. Then $Q(\rho) = \lceil \log_2 \prank(P) \rceil$.
\end{Thm}
\begin{proof}
We will first show $\Q(\rho) \leq \lceil \log_2 \prank(P) \rceil$. Let $r=\prank(P)$.  We will exhibit a purification $\ket{\psi}$ of $\rho$ with $\srank(\ket{\psi}) =r$. This combined with Lemma~\ref{lem:qrho} will show  $Q(\rho) \leq \lceil \log_2 r \rceil $. Let $C_x,D_y\in \mbC^{r\times r}$ be positive semi-definite matrices with $\tr(C_x D_y) = P(x,y)$, $\forall x \in [\dim(\h_A)], y\in [\dim_{\h_B}]$. For $i\in [r]$, let $\ket{v_x^i}$ be the $i$-th column of $\sqrt{C_x^T}$ and let $\ket{w_y^i}$ be the $i$-th column of $\sqrt{D_y}$. Define $\ket{\psi}$ in $\h_A \otimes \h_{A} \otimes \h_{A_1} \otimes \h_B \otimes \h_B \otimes \h_{B_1}$ as follows.
$$ \ket{\psi} \defeq \sum_{i=1}^r \left(\sum_{x \in [\dim(\h_A)]} \ket{x} \otimes \ket{x} \otimes \ket{v_x^i} \right) \otimes \left(\sum_{y \in [\dim(\h_B)]} \ket{y} \otimes \ket{y} \otimes \ket{w_y^i} \right) \enspace .  $$
It is clear that $\srank(\ket{\psi}) = r$. Also,
\begin{align*}
& \quad \   \lefteqn{\tr_{\h_A \otimes \h_{A_1} \otimes \h_B \otimes \h_{B_1}} \ket{\psi} \bra{\psi} }\\
& = \sum_{x\in  [\dim{\h_{A}}] \atop y\in  [\dim{\h_{B}}]} \ket{x}\bra{x} \otimes \ket{y}\bra{y} \left( \sum_{i,j=1}^r \qip{v_{x}^j}{v_x^i} \cdot \qip{w_{y}^j}{w_y^i}\right)   \\ 
& = \sum_{x\in  [\dim{\h_{A}}] \atop y\in  [\dim{\h_{B}}]} \ket{x}\bra{x} \otimes \ket{y}\bra{y}  \cdot \tr(C_x D_y)  \quad = \quad \rho \enspace . 
\end{align*}
Note that \alice and \bob after sharing $\ket{\psi}$ can either just output their first registers or measure their first registers in their respective computational bases to obtain $\rho$.
\vspace{0.1in}

Now we will show $\Q(\rho) \geq \lceil \log_2 \prank(P) \rceil$. Let $\ket{\psi} \in \h_A \otimes \h_{A_1} \otimes \h_B \otimes \h_{B_1}$ be a purification of $\rho$ with $\srank(\ket{\psi}) = r$ and $\Q(\rho) = \lceil \log_2 r \rceil$, as guaranteed by Lemma~\ref{lem:qrho}. We will show $r \geq \prank(P)$ and this will show the desired. Let
$$ \ket{\psi} =  \sum_{i=1}^r \left(\sum_{x\in  [\dim(\h_{A})]} \ket{x} \otimes \ket{v_x^i}\right) \otimes \left(\sum_{y\in  [\dim(\h_{B})]} \ket{y} \otimes \ket{w_y^i}\right) \enspace .$$
For all $x \in [\dim{\h_A}]$, define $r \times r$ matrices $C_x$ such that $C_x(j,i) = \qip{v_x^j}{v_x^i}$ for all $i,j \in [r]$. Similarly  for all $y \in [\dim(\h_B)]$, define $r \times r$ matrices $D_y$ such that $D_y(i,j) = \qip{w_y^j}{w_y^i}$ for all $i,j \in [r]$. From Fact~\ref{fact:psd}, $C_x, D_y$ are positive semi-definite  for all $x \in [\dim(\h_A)]$ and for all $y \in [\dim(\h_B)]$.
Consider
\begin{align*}
 \rho & = \tr_{\h_{A_1} \otimes \h_{B_1}} \ket{\psi} \bra{\psi}  \\
& =  \sum_{x\in  [\dim(\h_{A})] \atop y\in  [\dim(\h_{B})]} \ket{x}\bra{x} \otimes \ket{y}\bra{y} \left( \sum_{i,j=1}^r \qip{v_{x}^j}{v_x^i} \cdot \qip{w_{y}^j}{w_y^i}\right)   \\ 
& =  \sum_{x\in  [\dim(\h_{A})] \atop y\in  [\dim(\h_{B})]} \ket{x}\bra{x} \otimes \ket{y}\bra{y} \cdot \tr (C_x D_y)  \enspace .
\end{align*}
Therefore for all $x \in [\dim(\h_A)]$ and for all $y \in [\dim(\h_B)]$ we have $p_{x,y} = P(x,y) = \tr (C_x D_y)$. Hence $\prank(P) \leq r$.
\end{proof}

\suppress{
\section{Classical complexity of approximating a distribution} \label{sec:class}
In this section we prove Theorem~\ref{res:class}. We restate it here for convenience.
\begin{Thm} Let $\beta, \delta, \gamma >0$. Let $P$ be a distribution on  $\mcX \times \mcY$. Let random variables $XY$ be distributed according to $P$ such that $X$ is distributed in $\mcX$ and $Y$ is distributed in $\mcY$. Then,
    $$\R_{6\beta + \delta + 2\gamma}(X,Y) \leq C(X:Y)/\beta + 2\log(1/\delta) + \log\log(1/\gamma). $$ 
Here $C(X:Y)$ refers to the common information of $X$ and $Y$ (Definition~\ref{def:comminf}).
\end{Thm}
\begin{proof}
    Let $W$ be a random variable, taking values in the set $\mcW$, that achieves the minimum in the definition of $C(X,Y)$. Let $Z \defeq XY$ and $\mcZ \defeq \mcX \times \mcY$. The idea is to find a small number (in terms of $C(X,Y)$) of $w_i \in \mcW$'s such that the uniform average, over $w_i$s, of the conditional distributions $(Z|w_i)$ is close to $Z$ itself. Let $Q$ be the joint distribution of $(Z,W)$. Let
    \[
    k \defeq  I(Z:W) = \sum_{(z,w) \in \mcZ \times \mcW} Q(z,w) \log \frac{Q(z|w)}{Q(z)} .
    \]
Define
$$ \good \defeq \{(z,w) \in \mcZ \times \mcW ~|~ \frac{Q(z|w)}{Q(z)} \leq 2^{k/\beta} \}.$$
By Markov's inequality, $ q \defeq Q(\good) \geq 1-\beta$.     Define distribution $Q'$ on $\mcZ \times \mcW$ as follows. 
		\[
		Q'(z,w) = \begin{cases} \frac{1}{q} Q(z,w) & \text{if } (z,w) \in \good \\ 0 & \text{otherwise}\end{cases}.
		\] 
		Then
    \begin{align}
        \|Q'-Q\|_1 & = \frac{1}{2} \sum_{(z,w)\in \mcZ \times \mcW} \big|Q'(z,w) - Q(z,w)\big| \nonumber \\
        & = \frac{1}{2}  \left(\sum_{(z,w)\in \good} Q(z,w) (\frac{1}{q}-1) + \sum_{(z,w)\notin \good} Q(z,w) \right) \nonumber \\
        & = (1-q) \leq \beta .
    \end{align}
Note that for all $(z,w) \in \good$
$$ \frac{Q'(z,w)}{Q'(z)}{Q'(w)} \leq \frac{1}{\beta} $$
    Let us sample $\{w_1, w_2, \ldots, w_n\}$, where $n \defeq O(2^{k/\beta} \cdot \log(1/\gamma) \cdot \delta^{-2})$, where each $w_i \in \mcW$ is sampled according to the marginal distribution of $Q'$ on $\mcW$. Now for all $z\in \mcZ$, $\av_{w_i\leftarrow Q'} \left[ \frac{Q'(z|w_i)}{Q'(z)}\right] = 1$ and $Q'(z|w_i) \leq Q'(z)2^{k/\beta}$. Thus we are able to apply Chernoff's bound and get
    \begin{align}\label{eq:concentration}
        & \ \pr_{\{w_i\}} \Big[\Big|\frac{p'(z|w_1)+\cdots+p'(z|w_n)}{n} - p'(z)\Big| >\delta p'(z)\Big]  \\
        = & \ \pr_{\{w_i\}} \Big[\Big|\frac{p'(z|w_1)}{p'(z)2^{k/\beta}}+\cdots+\frac{p'(z|w_n)}{p'(z)2^{k/\beta}} - \frac{n}{2^{k/\beta}}\Big| > \frac{\delta n}{2^{k/\beta}}\Big] \\
        \leq & \ e^{-\delta^2n/(3\cdot2^{k/\beta})} = \gamma.
    \end{align}
		
		Now the protocol is that \alice and \bob sample an $i\in [n]$ uniformly at random and use $w_i$ to generate $X$ and $Y$, respectively. The resulting distribution is $p''(z) = (p(z|w_1) + \cdots + p(z|w_n)) / n$. We want to show the existence of $\{w_i\}$ to make this distribution close to the target distribution $(X,Y)$. Consider to take random $w_i$ as in the above analysis. Then
    \begin{align}
        &\quad \ \av_{\{w_i\}} \left[\sum_z |p''(z) - p(z)| \right] \\
				& \leq \av_{\{w_i\}} \left[ \sum_z \Big| \frac{\sum_i p(z|w_i)}{n} - p(z) \Big|\right] \\
				& \leq \av_{\{w_i\}} \left[\sum_z \Big(\Big| \frac{\sum_i p(z|w_i)}{n} - \frac{\sum_i p'(z|w_i)}{n} \Big| + \Big|\frac{\sum_i p'(z|w_i)}{n} - p'(z)\Big| + |p'(z) - p(z) | \Big) \right]\\
				& \leq \av_{w} \left[\sum_z | p(z|w) - p'(z|w) | \right]+ \sum_z\av_{\{w_i\}} \left[\Big|\frac{\sum_i p'(z|w_i)}{n} - p'(z)\Big| \right]  + \sum_z|p'(z) - p(z) | 
		\end{align}
		The first summand is 
		\begin{align}
			\sum_{w,z} | p(z,w) - p'(z|w)p(w) |  & \leq \sum_{w,z} | p(z,w) - p'(z,w) | + \sum_{w,z} p'(z|w)| p'(w) - p(w) | \\
		& \leq 2\beta + 2\beta = 4\beta,
		\end{align}
		by the upper bound of $\|p'-p\|_1$. The second summand is at most 
		\[
			(1-\gamma) \sum_z \delta p'(z) + \gamma\cdot 2 \leq \delta+2\gamma.
		\]
		by the bound following Eq.\eqref{eq:concentration}. The third summand is at most $2\beta$ again by the upper bound of $\|p'-p\|_1$. Putting everything together, we have $\av_{\{w_i\}}\left[\|p'' - p\|_1\right] \leq 6\beta + \delta + 2\gamma$. Thus there exist $\{w_i\}$ to satisfy the bound as well. The cost of the protocol is $\log n = k/\beta + 2\log(1/\delta) + \log\log(1/\gamma)$.				
\end{proof}
}
\suppress{
\section{Classical complexity of approximate a distribution}
In this section, we give an upper bound of $\R_\epsilon(P)$, the classical complexity of approximate a distribution, in terms of the common information introduced by Wyner \cite{Wyn75}.
\begin{Def}[Wyner, \cite{Wyn75}]
    The common information $C(X,Y)$ between two discrete dependent random variables $X$ and $Y$ is the minimum value of $I(XY,W)$, where the minimum is over all auxiliary random variable $W$ s.t. $X$ and $Y$ are independent conditioned on $W$.
\end{Def}
\begin{Thm}
    $\R_{\epsilon}(X,Y) \leq 2C(X,Y)/\epsilon + O(1)$.
\end{Thm}
\begin{proof}
    Suppose $W$ is the random variable achieving the minimum in the definition of $C(X,Y)$. Use $Z$ to denote $XY$ for notational convenience. The idea is to find a small (in terms of $C(X,Y)$) number of $w_i$'s \st the average of the conditional distributions $(Z|w_i)$ is close to $Z$ itself.

    First, denote by $p$ the distribution of $(Z,W)$. Assume that
    \[
    k = C(X,Y) = I(Z,W) = \sum_{z,w} p(z,w) \log p(z|w)/p(z).
    \]
    By Markov's inequality,
    \[
    \sum\Big\{p(z,w): \frac{p(z|w)}{p(z)} \leq 2^{k/\beta}\Big\} \geq 1-\beta.
    \]
    Let the set $Good_1$ to contain those $(z,w)$ \st $p(z|w)/p(z) \leq 2^{k/\beta}$. Define new random variables $(Z',W') = (Z,W)|Good_1$, and use $p'$ as the density function for $(Z',W')$. Then
    \begin{align}
        \|(Z',W') - (Z,W)\|_1 & = \sum_{z,w} \big|\pr[Z'=z, W'=w] - \pr[Z=z,W=w]\big| \\
        & = \sum_{(z,w)\in Good_1} \pr[Z=z,W=w] (1/p(Good_1)-1) + \sum_{(z,w)\notin Good_1} p(z,w)\\
        & \leq (1/(1-\beta)-1) + \beta = (2\beta-\beta^2)/(1-\beta) \approx 2\beta
    \end{align}
    Now consider to sample from $n = O(2^{k/\beta} \cdot \log(1/\gamma) \cdot \delta^{-2})$ $w_i$'s, and use Chernoff's bound, we have that for any fixed $z$,
    \begin{align}
        & \ \pr_{\{w_i\}} \Big[\Big|\frac{p'(z|w_1)+\cdots+p'(z|w_n)}{n} - p'(z)\Big| >\delta p'(z)\Big]  \\
        = & \ \pr_{\{w_i\}} \Big[\Big|\frac{p'(z|w_1)}{p'(z)2^{k/\beta}}+\cdots+\frac{p'(z|w_n)}{p'(z)2^{k/\beta}} - \frac{n}{2^{k/\beta}}\Big| > \frac{\delta n}{2^{k/\beta}}\Big] \\
        \leq & \ e^{-\delta^2n/(3\cdot2^{k/\beta})} = \gamma
    \end{align}
    Taking average over $z$ according to $p'$, we have
    \[
    \sum_v p'(v) \pr_{\{w_i\}} \Big[\Big|\frac{p'(z|w_1)+\cdots+p'(z|w_n)}{n} - p'(z)\Big| >\delta p'(z)\Big] \leq \gamma.
    \]
    Switching the probability and summation gives
    \[
    \av_{\{w_i\}} \sum_v p'(v)\mathbf{1}\big[\big|\frac{p'(z|w_1)+\cdots+p'(z|w_n)}{n} - p'(z)\big| >\delta p'(z)\big] \leq \gamma
    \]
    where $\mathbf{1}[\cdot]$ is the indictor function (taking 1 if the event is true and 0 otherwise). Thus there exist choices of $\{w_i\}$ to make the above summation at most $\gamma$. Let the set $Good_2$ to contain those $z$ with $|(p'(z|w_1)+\cdots+p'(z|w_n))/n - p'(z)| \leq \delta p'(z)$. Now the protocol is that \alice and \bob sample an $i\in [n]$ uniformly at random and generate $X$ and $Y$, respectively, using $w_i$. The resulting distribution for $(X,Y)$ is $p''(z) = (p'(z|w_1) + \cdots + p'(z|w_n)) / n$. Note that
    \begin{align}
        \|p'' - p'\| & = \sum_z |p''(z) - p'(z)| \leq \sum_{z\in Good_2} \delta\cdot p'(z) + \sum_{z\notin Good_2} |p''(z) - p'(z)| \\
        & \leq \delta + \sum_{z\notin Good_2} p'(z) + \sum_{z\notin Good_2} p''(z) \\
        & \leq \delta + \gamma + (1-(1-\delta)(1-\gamma)) \\
        & \leq 2(\delta+\gamma)
    \end{align}
    where the bound for $\sum_{z\notin Good_2} p''(z)$ in the second last inequality is by considering the complement, \ie $z\in Good_2$, where $|p''(z) - p'(z)| \leq \delta p'(z)$.

    Therefore, with the cost of $\log(n) = C(X,Y)/\beta + 2\log(1/\delta) + \log\log(1/\gamma)+O(1)$, \alice and \bob get $Z'' = (X'',Y'')$ with the distance $\|Z'' - Z\| \leq \|Z'' - Z'\| + \|Z' - Z\| \leq 2(\beta+\delta+\gamma)$. Arranging the parameters to make $2(\beta+\delta+\gamma) = \epsilon$ gives the conclusion.
\end{proof}
}

\suppress{
\section{Open problems}\label{sec:open}
\begin{enumerate}
    \item Can we give lower bounds for the approximate psd-rank? It has an application for an open question raised in \cite{ASTS+03} (page 2, paragraph -2). For the matrix where rows and columns are indexed by $\sqrt{n}$-subsets of $[n]$, consider the distribution $P$ uniformly on disjoint subsets. \cite{ASTS+03} showed $\Q_\epsilon(P) = O(\log n \log(1/\epsilon))$, and it asked about its optimality. Our result in last section transforms this question to bounds of approximate psd-rank. Note that the mutual information is not enough: It's not hard to calculate that $I(X,Y) = \Theta(1)$.
    \item For distribution $P$, we've given a nontrivial upper bound for $\R_\epsilon(P)$, which can also be viewed as an upper bound of the \emph{approximate} nonnegative rank. Can we say something for $\Q_\epsilon(P)$ (and thus approximate positive semi-definite  rank)?
\end{enumerate}
}

\subsection*{Acknowledgments}
R.J. is supported by the internal grants of Centre for Quantum Technologies. Y.S. and S.Z. were partially supported by China Basic Research Grant 2011CBA00300 (sub-project 2011CBA00301) and 2007CB807900 (sub-project 2007CB807901). Y.S was also supported in part by US NSF (1017335). Z.W. would
like to acknowledge the WBS grant under contract no. R-710-000-007-271. S.Z. was also supported by Research Grants Council of the Hong Kong S.A.R. (Project no. CUHK419309, CUHK418710, CUHK419011).

\bibliography{QStateGen}

\begin{thebibliography}{1}

\bibitem{ASTS+03}
Andris Ambainis, Leonard Schulman, Amnon Ta-Shma, Umesh Vazirani, and Avi
  Wigderson.
\newblock The quantum communication complexity of sampling.
\newblock {\em SIAM Journal on Computing}, 32(6):1570--1585, 2003.

\bibitem{EY36}
Carl Eckart and Gale Young.
\newblock The approximation of one matrix by another of lower rank.
\newblock {\em Psychometrika}, 1(3):211--218, 1936.

\bibitem{FMP+12}
Samuel Fiorini, Serge Massar, Sebastian Pokutta, Hans~Raj Tiwary, and Ronald
  de~Wolf.
\newblock Linear vs. semidefinite extended formulations: Exponential separation
  and strong lower bounds.
\newblock In {\em Proceedings of the 44th ACM Symposium on Theory of
  Computing}, 2012.

\bibitem{NC00}
Michael Nielsen and Isaac Chuang.
\newblock {\em Quantum Computation and Quantum Information}.
\newblock Cambridge University Press, Cambridge, UK, 2000.

\bibitem{Zha12}
Shengyu Zhang.
\newblock Quantum strategic game theory.
\newblock In {\em Proceedings of the 3rd Innovations in Theoretical Computer
  Science}, pages 39--59, 2012.
\newblock Earlier at arXiv:1012.5141 and QIP'11.

\end{thebibliography}
\bibliographystyle{plain}


\suppress{
\section{Discussions in $\Q(\rho)$}

However, the minimization is not easy to solve. We desire more explicit characterizations or bounds.

It seems that even the zero-error case is not trivial. Some simple observations are as follows. For a bipartite pure state $\ket{\psi}$, let $\rank(\ket{\psi})$ denote its Schmidt rank. 

\vspace{.5em} \textbf{Upper bound 1}: $\Q(\rho) \leq \min\{\log_2 K + \max_i \log_2 \rank(\ket{\psi_i}): \rho = \sum_{i=1}^K p_i \ket{\psi_i}\bra{\psi_i}\}$. (\alice samples $p_i$ and sends $i$ to \bob, and then they generate $\ket{\psi_i}$. This bound is a simple generalization of the pure state case.)

Note that the decomposition in the upper bound does matter. For example, the maximum mixed states of two qubits $I/4$ can be decomposed in the computational bases, giving an upper bound of $\log_2 4 = 2$. But it can also be decomposed as the four EPR pairs each with probability 1/4, then the upper bound becomes $\log_2 4 + \log_2 2 = 3$.

One may wonder about the performance of spectral decomposition. It is actually not necessarily the best. For example, if the state is
\begin{equation}
    \rho = \frac{1}{2}\cdot\frac{\ket{00}+\ket{11}}{\sqrt{2}}\frac{\bra{00}+\bra{11}}{\sqrt{2}} + \frac{1}{2}\cdot\frac{\ket{11}+\ket{22}}{\sqrt{2}}\frac{\bra{11}+\bra{22}}{\sqrt{2}},
\label{eq:eg1}
\end{equation}
then it has spectral decomposition $\rho = 1/4\ket{\psi_1}\bra{\psi_1} + 3/4\ket{\psi_2}\bra{\psi_2}$, where $\ket{\psi_1} = (\ket{22}-\ket{00})/\sqrt{2}$ and $\ket{\psi_2} = (\ket{00}+2\ket{11}+\ket{22})/\sqrt{6}$. This decomposition gives a worse upper bound than the Eq. \eqref{eq:eg1}.

\vspace{.5em} \textbf{Upper bound 2}: $\Q(\rho) \leq \min\{\Q(p): \rho = \sum_{x,y} p_{x,y} (\rho_x^A\otimes \rho_y^B)$
. (\alice and \bob samples $p_{x,y}$ and then they generate $\rho_x^A$ and $\rho_y^B$ respectively. This bound is a simple generalization of the classical distribution case, but it may also provide nothing since if $\rho$ is entangled, then it cannot be decomposed in this way.)


\vspace{.5em} \textbf{Lower bound 1}: $\Q(\rho) \geq \max\{\Q(p): p$ is obtained from some local measurement on $\rho\}$. (If they get $\rho$, then they can apply any local measurement and get a distribution $p$.)

The lower bound is tight when $\rho$ is a classical distribution or a pure quantum state. The former is trivial. The latter: Suppose $\rho = \ket{\psi}\bra{\psi}$ where $\ket{\psi} = \sum_{i=1}^r \sigma_i\ket{\psi_i^A}\ket{\psi_i^B}$ with $\sigma_i > 0$. 
\alice takes the measurement on the basis $\ket{\psi_i^A}$ and \bob on $\ket{\psi_i^B}$, then $p(i,j) = \sum_{ij}\delta_{ij}\sigma_i^2$. So $P = diag(\sigma_1^2, ..., \sigma_r^2, 0, ..., 0)$, thus $\Q(p) \geq \log_2 \rank(P) = \log_2 r = \log_2 \srank(\ket{\psi}) = \Q(\ket{\psi})$.

\vspace{.5em} \textbf{Question 1}: What's the real answer? Is it actually true that the lower bound is tight?
}

\end{document}